\documentclass[12pt]{article}
\usepackage{fullpage, amsmath, amsthm, amssymb}
\usepackage{amsfonts}
\usepackage{amssymb}

\DeclareMathOperator{\poly}{poly}
\DeclareMathOperator{\cc}{cc}
\DeclareMathOperator{\secur}{sec}

\newtheorem{theorem}{Theorem}

\newtheorem{claim}[theorem]{Claim}

\newtheorem{conjecture}[theorem]{Conjecture}
\newtheorem{corollary}[theorem]{Corollary}

\newtheorem{lemma}[theorem]{Lemma}

\begin{document}

\title{Multitask Efficiencies in the Decision Tree Model}
\author{Andrew Drucker\thanks{Email: adrucker@mit.edu. \ This work was supported by a Kunzel Fellowship while at UC San Diego (2006-7), by Scott Aaronson while the author was a visiting student at MIT (2007-8), and by an Akamai Presidential Graduate Fellowship at MIT (2008-present). }\\MIT}
\date{}
\maketitle

\begin{abstract}
In Direct Sum problems \cite{KRW}, one tries to show that for a given computational model, the complexity of computing a collection $F = \{f_1(x_1), \ldots f_l(x_l)\}$ of finite functions on independent inputs is approximately the sum of their individual complexities.  In this paper, by contrast, we study the diversity of ways in which the joint computational complexity can behave when all the $f_i$ are evaluated on a \textit{common} input.  We focus on the deterministic decision tree model, with depth as the complexity measure; in this model we prove a result to the effect that the `obvious' constraints on joint computational complexity are essentially the only ones.

The proof uses an intriguing new type of cryptographic data structure called a `mystery bin' which we construct using a small polynomial separation between deterministic and unambiguous query complexity shown by Savick\'{y}.  We also pose a variant of the Direct Sum Conjecture of \cite{KRW} which, if proved for a single family of functions, could yield an analogous result for models such as the communication model.

\end{abstract}

\section{Introduction\label{INTRO}}

A famous `textbook' result in algorithms \cite{Poh}, \cite{CLRS00} is that given a list of $n$ integers, it is possible to locate both the maximal and the minimal element using $\lceil \frac{3n}{2} \rceil - 2$ comparisons; this is an improvement over the na\"{\i}ve strategy of computing each separately, which would take $2n - 2$ comparisons.  In such a setting we say, informally, that a `multitask efficiency' exists between the MAX and MIN functions, because the tasks of computing them can be profitably combined.  We emphasize that, in contrast to `direct sum problems', where we wish to understand the complexity of computing several functions (usually the same function) on several disjoint input variable sets, here we are interested in evaluating multiple functions on a \textit{common input}.

While in relatively simple examples such as MAX/MIN the multitask efficiencies that exist can be well-understood, for even slightly richer examples the situation becomes more complex.  For instance, in the `set maxima' problem, one is given a list of $n$ integers and a family of subsets $S_i$ of $[n]$, and is asked to find the maximal element in each corresponding subset of the list.  A significant amount of research has focused on finding upper- and lower-bounds on the number of comparisons needed, as determined by the set-family structure; see, e.g., \cite{GKS}.

There are even more mysterious examples of multitask efficiencies in computation.  Given a linear map $T(x): \mathbb{F}_2^n \rightarrow \mathbb{F}_2^n$, one can ask about the complexity of circuits computing $T$, composed of $\mathbb{F}_2$-addition gates.  Clearly every \textit{individual} output coordinate can be computed by a bounded-fanin linear circuit of size $n$ and depth $\lceil\log n \rceil$, but to compute all $n$ outputs simultaneously may require greater resources.  However, it has been open for over 30 years to provide an explicit family of maps which provably cannot be computed by linear-size, logarithmic-depth linear circuits, even though such maps are known to exist in abundance \cite{Val77}.

\subsection{Our Results}

Inspired by these examples, in this paper we propose and begin a systematic study of computational models from the point of view of the multitask efficiencies they exhibit.  Formally we approach this in the following way.  We fix a computational model $M$ capable of producing output over any finite alphabet, and a notion of cost for that model (worst-case number of comparisons, decision tree depth, etc.).  Given a collection $F = \{f_1(x), f_2(x), \ldots f_l(x)\}$ of functions on a common input $x \in \{0, 1\}^n$, define the \textit{multitask cost function} $C_F(X): \{0, 1\}^l \rightarrow \mathbb{R}$ by letting $C_F(X)$ equal the minimum cost of any algorithm in $M$ that, on input $x \in \{0, 1\}^n$, outputs in some specified order the values $f_i(x)$, for every $i$ such that $X_i = 1$.  (We use capitalized variable names for vectors that index subsets of a function family $F$, to distinguish them from the lower-case vectors $x$ which will denote inputs to $F$.)

The question we are interested in is this: What kind of functions $C_F(X)$ can arise in this way, when we range over all choices of $F$?

There are some obvious constraints on $C_F$.  For many reasonable definitions of cost, $C_F$ will be \textbf{nonnegative and integer-valued} (at least for worst-case notions of cost, which we will always be considering).  As long as the functions in $F$ are non-constant (and we will assume this throughout), $C_F(X)$ will be 0 if and only if $X = \mathbf{0}$.  

We expect $C_F$ to be \textbf{monotone} (but not necessarily strictly monotone), since any algorithm computing a subset $S \subseteq F$ of functions can be trivially modified to compute any $S' \subset S$.  Finally, $C_F$ should be \textbf{subadditive}; that is, we should always have $C_F(X \vee Y) \leq C_F(X) + C_F(Y)$.  This is because an algorithm can always solve two subcollections of functions separately and then combine the results in its output. 

Are there any other constraints?  We now illustrate by example that, for at least some models of computation, there are functions $C(X)$ obeying the constraints above, which do not correspond to $C_{F}(X)$ for any choice of collection $F$.  We consider the deterministic  decision tree model, with depth as the complexity measure.

For $X \in \{0, 1\}^3$, let $||X||$ be the Hamming weight of $X$, and define

\begin{equation*}
C^*(X)= 
\begin{cases} 0 & \text{if $||X|| = 0$,}
\\
1 & \text{if $||X|| \in \{1, 2\}$,}
\\
2 &\text{if $||X|| = 3$.}
\end{cases}
\end{equation*}
One can verify that $C^*(X)$ satisfies nonnegativity, monotonicity, and subadditivity. Now suppose for contradiction's sake that some family $F = \{f_1(x), f_2(x), f_3(x)\}$ satisfies $C_{F}(X) = C^*(X)$ for all $X$. This means that any two functions in $F$ can be computed with one query to $x$, while it requires 2 queries to compute all three.

Since $C^*(1, 1, 0) = 1$, $f_1$ and $f_2$ must depend only on a single shared input bit $x_i$.  Similarly $C^*(1, 0, 1) = 1$ implies that $f_1, f_3$ each depend on a single shared input bit $x_j$, so $i = j$.  But then a single query to $x_i$ determines all three functions, so that $C_{F}(1, 1, 1) = 1 \neq C^*(1, 1, 1)$.  This contradicts our assumption.

The example of $C^*$ suggests that other significant constraints might exist on multitask cost functions for decision-tree complexity.  However, we will show that there is a strong sense in which this is false.  In Section \ref{ecfs} we formally define \textit{economic cost functions} as functions obeying nonnegativity (strict except at $\mathbf{0}$), monotonicity, and subadditivity; the rest of the paper is then devoted to proving the following result:

\begin{theorem}\label{main}  Given any collection 
\[F = \{f_1(x), f_2(x), \ldots f_l(x)\}\]
of nonconstant Boolean functions, $C_F(X)$ (defined relative to the adaptive query model) is an economic cost function.  

Furthermore, given any economic cost function $C(X): \{0, 1\}^l \rightarrow \mathbb{Z}$, and an $\epsilon > 0$, there exist integers $n$, $T > 0$, and a collection $F = \{f_1(x), \ldots f_l(x)\}$ of (total) Boolean functions on a common $n$-bit input $x$, such that, for all $X$, 
\[(1 - \epsilon)T\cdot C(X) \leq C_F(X) \leq (1 + \epsilon)T\cdot C(X).\]

\end{theorem}
That is, there exist multitask cost functions $C_F(X)$ to approximate any economic cost function, if that economic cost function is allowed to be `scaled up' by a multiplicative factor and we allow a multiplicative error of $\epsilon > 0$.  Theorem \ref{main} would remain true if we allowed economic cost functions to take non-integral values, since (up to a scaling factor) such functions can be arbitrarily well-approximated by integral economic cost functions.

We summarize Theorem \ref{main} by saying that the adaptive query model is \textit{universal for economic cost functions}.  For any model $M$ of computation with an associated notion of cost, we say that $M$ is universal for economic cost functions if the analogue of Theorem \ref{main} is true with multitask cost functions from $M$ replacing those of the adaptive query model.

As a consequence of Theorem \ref{main}, we obtain a universality result for any deterministic, adaptive model into which we can `embed the query model'.  For example of what we mean, let us consider the comparison model of computation over lists of integers in which a basic step is a comparison of two list elements (as used in the MAX/MIN result mentioned earlier). 

Let $F = \{f_1(x), \ldots f_l(x)\}$ be any collection of Boolean functions with domain $\{0, 1\}^n$.  Based on $F$, we define a collection $G = \{g_1(a), \ldots g_l(a)\}$ of Boolean-valued functions $g_j(a)$ taking as common input a list of $2n$ integers $a = (a_1, \ldots a_{2n})$.  First, let $b_i = b_i(a)$ be an indicator variable for the event $[a_{2i - 1} < a_{2i}]$.  Then define
\[g_j(a) = f_j(b_1, \ldots b_n).\]
The values $b_i$ are each computable by a single comparison, and each pair $b_i, b_{i'}$ are functions of disjoint variable-sets, so we see that the cost of computing any subcollection of $G$ on a common input is exactly the cost (in the Boolean adaptive-query model) of computing the corresponding subcollection of $F$.  

Since the query model thus `embeds' into the comparison model (and since cost functions in the comparison model can be easily seen to be economic cost functions), in light of Theorem \ref{main} we conclude:

\begin{corollary}\label{comparison} The comparison model is universal for economic cost functions.  \qed

\end{corollary}

Proving such a result in the communication model seems difficult, and would require a better understanding of the Direct Sum phenomenon.  We next state a `Query-Model Embedding Conjecture' that would suffice to prove that the communication model is universal for economic cost functions, along the lines of Corollary \ref{comparison}.

Let $n, k > 0$ be integers.  Given $f(x, y): \{0, 1\}^{2n} \rightarrow \{0, 1\}$, and a function $g(z): \{0, 1\}^{k} \rightarrow \mathbb{N}$, define a function $(g \circ f): \{0, 1\}^{2nk} \rightarrow \mathbb{N}$ by 
\[(g \circ f)(x_1, y_1, x_2, y_2, \ldots x_k, y_k) = g(f(x_1, y_1,), \ldots f(x_k, y_k)). \]
In the communication problem for $g \circ f$ we understand Alice to receive all $x$-inputs and Bob all $y$-inputs.  Let $\cc(h)$ denote the (adaptive, deterministic) communication complexity of computing the ($\mathbb{N}$-valued) function $h$, by a protocol in which Alice speaks first, and both parties learn the function value. As usual let $D(g)$ denote the decision tree complexity of computing $g$.

\begin{conjecture} For every $k\in \mathbb{N}$ and $\delta \in (0, 1)$, there exists $n> 0$ and a function $f: \{0, 1\}^{2n} \rightarrow \{0, 1\}$ (with $\cc(f) > 0$) such that for all $g: \{0, 1\}^{k} \rightarrow \mathbb{N}$, we have

\[\cc(g \circ f) \geq (1 - \delta)\cc(f) D(g).  \]

\end{conjecture}

We can show a nearly matching upper bound 
\[\cc(g\circ f) \leq \cc(f)D(g)\]
for all choices of $g, f$, by the following protocol idea: The players consider $\{b_i := f_i(x_i, y_i)\}_{i \leq k}$ as bits to be `queried', and simulate an optimal decision tree on these bits; whenever they want to determine some $b_j$, they execute the optimal communication protocol for $f$ on $(x_j, y_j)$.  This makes them both learn $f(x_j, y_j)$, so they both know which bit $b_i$ is to be `queried' next.

Note that the conjecture asserts a strengthened form of the Direct Sum property, for some particular family of functions $f$: by setting $g$ to be a function that outputs an encoding of its input, we see that computing $f(x, y)$ on $k$ independent input pairs requires nearly $k$ times as much communication as for one pair.  

Unable to prove the conjecture, we can at least note the following: the Conjecture really is sensitive to our choice of `inner' function $f$.  For example, let $f(x, y) = x \vee y$, and let $g$ be the OR function on $k$ bits.  Then the communication complexity of computing $ (g \circ f) = \bigvee_{i = 1}^k (x_i \vee y_i) = (\bigvee_{i = 1}^k x_i)  \vee (\bigvee_{i = 1}^k y_i)$ is $O(1)$, even though each $f(x_i, y_i)$ has nonzero communication complexity and the OR$_k$ function has decision tree complexity $k$.

We suspect that a \textit{random} function $f(x, y)$, on an input size sufficiently large compared to $k$ and $\frac{1}{\delta}$, should be a suitable inner function.

Our conjecture also appears somewhat related to the Enumeration and Elimination Conjectures of \cite{ABGKT} (so far unresolved).  These are another type of variant of the Direct Sum Conjecture of \cite{KRW}.  We are not, however, aware of any formal implication between these conjectures and the ours.

\subsection{Outline and Methods}

To prove Theorem \ref{main}, a key tool is the notion of hitting sets of weighted set systems.  Given a set family $\mathbf{A} = \{ A_1, \ldots A_l\}$ over a universe $U$ and a weight function $w: U \rightarrow \mathbb{R}$, the weight of a subset $B \subseteq U$ is defined as the sum of $B$'s members' weights.  $B$ is called a \textit{hitting set} for a subfamily $S \subseteq \mathbf{A}$ if $B$ intersects each $A_i \in S$.  The \textit{hitting-set cost function} $C_{\mathbf{A}}(X): \{0, 1\}^l  \rightarrow \mathbb{N}$ gives the minimum weight of any $B$ that is a hitting set for $S_X = \{A_i: X_i = 1\}$.

We use these notions to derive a useful representation lemma (Lemma \ref{hsuniv}): for any economic cost function $C(X)$ on $l$ bits, there exists a family $\mathbf{A} = \{ A_1, \ldots A_l\}$ over a weighted universe $(U, w)$, such that $C_{\mathbf{A}}(X) = C(X)$.  The (simple) proof of Lemma \ref{hsuniv} is given in Section \ref{setsys}.

As a concrete example to illustrate the expressive power of these hitting-set cost functions, we present a simple weighted set system whose hitting-set cost function is exactly the example function $C^*(X)$ presented earlier in the Introduction.  (This will not be the set system that would be produced by our general method.)  Let $\mathbf{A}$ be the family of all 2-element sets over the universe $U = \{u_1, u_2, u_3\}$ (so, $|\mathbf{A}| = 3$), and let $w(u_i) = 1$, for each $u_i \in U$.  Note that any one or two of the sets from $\mathbf{A}$ has a hitting set of size 1, but to hit all of $\mathbf{A}$ requires two elements.  Since each element has unit weight, $C_{\mathbf{A}}(X)$ is exactly $C^*(X)$.

Returning now to the discussion of our main strategy, it will suffice to solve the following problem: given a weighted set system $\mathbf{A} = \{A_1, A_2, \ldots A_l\}$, produce a collection $F = \{f_1, \ldots f_l\}$ of Boolean functions over some domain $\{0, 1\}^n$ such that $C_F(X)$ is approximately a multiple of $C_{\mathbf{A}}(X)$.

Here is a high-level sketch of our collection $F$.  For each $u \in U$, we create a block $y_u$ of input variables called the `bin' for $u$; $x$ is the disjoint union of these blocks.  $y_u$ represents, in a carefully defined way, the contents of a conceptual `bin' which contains at most one `key' $k$ from a large set $K$ called the `keyspace'.

The bin representations and a value $T > 0$ are chosen in such a way that the following (informal) conditions hold:

\begin{itemize}
\item[(i)] The contents of any bin $y_u$ can be determined in at most $w(u) T$ queries;
\item[(ii)] For any fixed $k \in K$ and any bin $y_u$, it can be determined with `very few' queries whether $k$ is in the bin (so that this step is `essentially for free' in comparison to the queries described in (i));
\item[(iii)] If the number of queries an algorithm makes to the bin $y_u$ is even `noticeably' less than $w(u) T$, the amount of information it gains about the bin contents is `tiny', that is, the data seen is consistent with almost any $k \in K$ occupying the bin.  (At least, this outcome is unavoidable when an appropriately chosen adversary strategy determines the answers to queries as they are made.)
\end{itemize}

We will formalize bins obeying the above properties in the notion of `mystery bin functions' in Section \ref{mbinsec}.

Returning to the sketch construction of our function collection, for $i \in \{1, 2, \ldots l\}$, define $f_i(x) = 1$ iff there exists some $k \in K$ that is contained in each of the `mystery bins' $y_u$ corresponding to elements $u \in A_i$.

To informally analyze this collection, fix any nonzero $X \in \{0, 1\}^l$, indexing a subcollection $S_X \subseteq \mathbf{A}$.

For an upper bound on $C_F(X)$, pick a minimal-weight hitting set $B$ for $S_X$, so $w(B) = C_\mathbf{A}(X)$.  In the first phase, for each $u \in B$, let our algorithm determine the bin contents of $y_u$.  By property (i) this phase uses at most $ w(B) T$ queries.

Next comes the second phase.  For every $A_i \in S_X$, there's a $u \in A_i \cap B$, whose bin contents we've determined; if the bin $y_u$ was empty we can conclude $f_i(x) = 0$.  If the bin contained the element $k \in K$ (remember that at most one key lies in each bin), query the bins of all other elements $u' \in A_i$ to see if $k$ is in all of them.  If so, $f_i(x) = 1$, otherwise $f_i(x) = 0$.  

Thus our algorithm succeeds in computing $\{f_i(x): X_i = 1\}$.  By property (ii) above, the query complexity of this second phase is `negligible', giving $C_F(X) \leq (1 + \epsilon)T\cdot C_\mathbf{A}(X)$ as needed.

For the lower bound, we pit any algorithm using fewer than $(1 - \epsilon )T\cdot C_\mathbf{A}(X)$ queries against an adversary strategy that runs the adversary strategies for each mystery bin in parallel.  Since $C_{\mathbf{A}}(X)$ is the minimal cost of any hitting set for $S_X$, at the end of this run of the algorithm there must exist some $A_i \in S_X$ such that for each $u \in A_i$, $y_u$ receives noticeably less than $ w(u) T$ queries.  Using property (iii) of mystery bins, we then argue that the algorithm fails to determine the value $f_i(x)$.  This will prove $C_F(X) \geq (1 - \epsilon )T\cdot C_\mathbf{A}(X)$.

The main technical challenge in implementing the above idea is to design the right representation of the bin contents of the blocks $y_u$ to guarantee the `mystery bin' properties.  To build mystery bin functions, we will exploit a small polynomial separation between decision tree depth and unambiguous certificate complexity, due to Savick\'{y} \cite{Sav02}.  We describe his result, and reformulate it for our purposes, in Section \ref{tusps}. 

How does Savick\'{y}'s result facilitate our construction of `mystery bins'?  Roughly speaking, the gap between deterministic and circuit complexity in his theorem yields the query-complexity gap between properties (i) and (ii) of mystery bins, while the key contribution of unambiguity is in allowing us to construct mystery bin functions in which the bin always contains at most one key.  In the algorithm described above to compute $\{f_i(x): X_i = 1\}$, this allows the query complexity of the second phase to remain negligible, yielding the upper bound we need on $C_F(X)$.

In the course of building mystery bin functions, another useful device called a `weak exposure-resilient function' is also introduced and used.  This object, an encoding method that looks uninformative when restricted to a small number of coordinates, is indeed a weak special case of the `exposure-resilient functions' studied in \cite{CDHKS}; however, the parameters we need are easily obtainable and so we provide a self-contained (probabilistic) construction and analysis.

\section{Definitions and Preliminary Results}

\subsection{Vectors and Economic Cost Functions} \label{ecfs}

Given two bitvectors $X = (X_1, \ldots X_l), Y = (Y_1, \ldots Y_l)$, we write $X \leq Y$ if $X_i \leq Y_i$, for all $i = 1, 2, \ldots n$.  We define the vector $Z = X \vee Y$ by the rule $Z_i = X_i \vee Y_i$.

Note that, in this paper, we use capital-letter variable names ($X, Y, Z$) to refer to vectors indexing `bundles of goods', and we use lower-case variable names to refer to other vectors, such as the inputs and outputs to functions whose decision-tree complexity we will analyze.

We say that a function $C(X): \{0, 1\}^l \rightarrow \mathbb{Z}$ is an \textit{economic cost function} if it satisfies the following conditions:
\begin{itemize}
\item[(1)] $C(X) \geq 0$, and $C(X) = 0 \Leftrightarrow X = \mathbf{0}$;
\item[(2)] For all $X, Y$, $X \leq Y$ implies $C(X) \leq C(Y)$;
\item[(3)] For all $X, Y$, $C(X \vee Y) \leq C(X) + C(Y)$.
\end{itemize}
We call such functions `economic cost functions' due to the following informal interpretation: consider the input $X \in \{0, 1\}^l$ to $C$ represent a certain subset of $l$ distinct `goods' that a company is capable of producing.  If $C(X)$ represents the cost to the company of producing one each of the goods indexed by the 1-entries of $X$, then intuitively, we expect $C$ to obey condition (1) because there's `no free lunch'.  Condition (2) supposes that, to produce one bundle of goods, one can always produce a larger bundle of goods and `throw away' the unwanted ones (and we assume free garbage disposal).  Condition (3) supposes that, to produce two (possibly overlapping) bundles $X, Y$ of goods, we can always separately produce the two bundles.  Equality may not always hold in condition (3), even for disjoint bundles of goods, due to possible `multitask efficiencies' arising in production.

We note in passing that the definition of economic cost functions is a special case of the more general notion of `outer measures' on lattices; see \cite{Bir}, Chapter 9.

\subsection{Decision Trees and Multitask Cost Functions}

We will consider decision trees taking Boolean input vectors but with outputs over a possibly non-Boolean alphabet.  A (deterministic, adaptive) \textit{decision tree} $T$ over the variables $x = x_1, \ldots x_n$ is a finite rooted binary tree whose internal nodes $u$ are each labeled with some variable index $i(u)$ and have designated `left' and `right' child nodes, and whose leaf (`output') nodes $l$ are each labeled with an element $v(l) \in B$, where $B$ is some finite alphabet.

A decision tree $T$ defines a function $f: \{0, 1\}^n \rightarrow B$ in the following way: given an input $x \in \{0, 1\}^n$, we begin at the root node.  Whenever we are at an internal node $u$, we look at the input variable $x_{i(u)}$.  If $x_{i(u)} = 0$, we move to the left child of $u$; if $x_{i(u)} = 1$, we move to the right child of $u$.  Eventually we arrive at an output node $l$, and we define $f(x) = v(l)$. 

We will be often consider decision trees $T$ whose output is a bitvector: $B = \{0, 1\}^l$ for some $l > 0$.  If the $i$th bit of $T$'s output is governed by the function $f_i(x)$, we say $T$ \textit{computes the collection} $\{f_1(x), f_2(x), \ldots f_s(x)\}$.

By the \textit{depth} of $T$, denoted $D(T)$, we mean the length of the longest path from the root in $T$, stepping exclusively from parent to child.  Given a collection of functions $S = \{f_1(x), f_2(x), \ldots f_l(x)\}$, we define the (deterministic, adaptive) \textit{query complexity} of $S$ as $D(S) =$ min $\{d:$ there exists a decision tree $T$ of depth $d$ computing the collection $S\}$.  If $S$ is a single function, $S = \{f\}$, we also write $D(f) = D(S)$. 

We next define, for any finite collection $F$ of functions, a function $C_F$ which summarizes the multitask efficiencies existing among the members of $F$ (relative to the decision-tree depth model of cost).  Given a collection $F$ of functions, $F = \{f_1(x), f_2(x), \ldots f_l(x)\}$ on a common input, we define the \textit{multitask cost function} $C_F(X): \{0, 1\}^l \rightarrow \mathbb{Z}$ \textit{associated with $F$} by $C_F(X) = D(S_X)$, where $f_i \in S_X \Leftrightarrow X_i = 1$.  We define $C_F(\mathbf{0}) = 0$.

Thus $C_F(X)$ gives the `cost' of certain `bundles of goods', where cost is interpreted as decision tree depth, and the different `bundles of goods' in question are the various subcollections of functions from $F$.  As promised by part of Theorem \ref{main}, we will show (Lemma \ref{easydir}) that for any $F$, $C_F(X)$ is always an economic cost function as defined in Section \ref{ecfs}.

\subsection{Search Problems and TUSPs} \label{tusps}

Although in this paper we are primarily interested in the query complexity of (collections of) decision problems, our proof techniques also involve \textit{search problems} (in the query model), defined next.  Our definitions and terminology will be slightly idiosyncratic, but for the most part could be altered slightly to match up with definitions from \cite{LNNW}.

Say that a string $w \in \{0, 1, *\}^n$ \textit{agrees with} $x \in \{0, 1\}^n$ if for all $i \in [n]$, $w_i \in \{0, 1\}$ implies $w_i = x_i$.  A search problem on domain $\{0, 1\}^n$ is specified by a subset $W \subseteq \{0, 1, *\}^n$ called the `witnesses'.  We say that a decision tree $T$ solves the search problem $W$ if (i) for every input $x$ that agrees with at least one $w \in W$, $T(x)$ outputs some $w' \in W$ agreeing with $x$ (if there are more than one such $w'$, we don't care which one), and (ii) if $x$ agrees with no $w \in W$, $T(x)$ outputs `no match'.

Given a search problem $W$, let $s(W)$ denote the maximum number of 0/1 entries in any $w \in W$.  Write $D(W)$ to denote the minimum depth of any decision tree solving $W$.

$W$ is called a \textit{total} search problem if all $x \in \{0, 1\}^n$ agree with at least one $w \in W$.  $W$ is called a \textit{unique} search problem if all $x$ agree with at most one $w \in W$.  In this paper we will deal with search problems $W$ that are both total and unique; we call such a $W$ a \textit{TUSP} for brevity.  A TUSP $W$ defines a (total, single-valued) function from $\{0, 1\}^n \rightarrow W$ mapping $x$ to the unique witness $w$ agreeing with $x$; we denote this function by $W(x)$.

For TUSPs $W$, as for other search problems, it is easy to see that $s(W) \leq D(W)$: for any decision tree $T$ solving $W$, the variables read by $T$ on an input $x$ must include all the $0/1$ entries in $w = W(x)$.  In fact, up to an at-most quadratic factor, this inequality is tight:

\begin{theorem}\label{quadratic} \cite{BI}, \cite{HH}, \cite{Tar89} For all unique search problems, $D(W) \leq s(W)^2$.  
\end{theorem}

\begin{proof} The proof is essentially identical to that of a related result, which states that decision-tree depth complexity is most the square of the `certificate complexity' for Boolean functions \cite{BI}, \cite{HH}, \cite{Tar89}.

Let $s = s(W)$.  We define a query algorithm as follows: on input $x$, proceed in phases.  At the beginning of phase $t$, let $W_t \subseteq W$ be the set of `live' witnesses, i.e. those that agree with the bits of $x$ seen so far.  Say that $i \in [n]$ is an `active' coordinate for $w \in W_t$ if $w_i \in \{0, 1\}$ and $x_i$ has not been queried.  In each phase $t$, the algorithm picks an arbitrary $w \in W_t$ and queries $x$ on each of the active coordinates $i$ for $w$.

Since $W$ is a unique search problem, every distinct $w, w' \in W_t$ disagree on at least one coordinate $i$ active for both $w$ and $w'$.  Thus, in each phase $t$ and for every $w \in V_t$, the number of active coordinates for $w$ decreases by at least one.  After at most $s$ phases, then, no live $w$ has any active coordinates; hence it either disagrees with $x$ on one of the bits already seen, or agrees with $x$ on each $i$ with $w_i \in \{0, 1\}$.  It follows that the decision tree for our algorithm solves $W$, while making at most $s + (s - 1) + \ldots + 1 \leq s^2$ queries. \end{proof}

In 2002 Petr Savick\'{y} \cite{Sav02} proved a theorem implying that, in general, $D(W)$ is not bounded by any constant multiple of $s(W)$ for TUSPs.  He uses different terminology and states a slightly different result than we need, so we will have to `unpack' his result a little.

A \textit{DNF formula $\psi$} is an OR of clauses, each of which consists of  the AND of one or more literals or negated literals.  Say that $\psi$ is an \textit{unambiguous DNF (uDNF)} if any input $x$ satisfies at most one of its clauses.  Savick\'{y} showed

\begin{theorem}\label{sav1}\cite{Sav02}  There exists a family of functions 
\[\{G_i: \{0, 1\}^{4^i} \rightarrow \{0, 1\} \}_{i \in \mathbb{N}} \quad{} \text{ such that}\]
 \begin{itemize}
\item[(i)] $G_i$ and $\overline{G_i}$ each have uDNF representations in which each clause has size at most $s_i = 3^i$;
\item[(ii)] $D(G_i) \geq \frac{4^i + 2}{3} = \Omega (s_i^{\gamma})$, where $\gamma = \log_3 (4) > 1$.
\end{itemize}
\end{theorem}

Theorem \ref{sav1} is very close, but not identical, to the combination of Theorems 3.1 and 3.6 from \cite{Sav02}.  That paper was concerned with the complexity measure $p(f)$ defined as the minimal number of clauses in any uDNF representation of $f$, whereas we are concerned with minimizing the maximum size of any clause as in Theorem  \ref{sav1}; also, Savick\'{y} lower-bounds the number of leaves of any decision tree for $f$ rather than its depth.  However, the particular function family \cite{Sav02} gives is seen by inspection (and noted by the author) to satisfy condition (i), while condition (ii) follows from Savick\'{y}'s lower bound on number of leaves in any decision tree computing $G_i$, after noting that a decision tree with $k$ leaves has depth at least $\lceil \log (k) \rceil$. this yields Theorem  \ref{sav1}.

We remark that it to prove our main theorem, we don't really need the full strength of Theorem \ref{sav1}.  Specifically, it would be enough that just \textit{one} of $G_i$ or $\overline{G_i}$ had uDNF representations with short clauses relative to the query complexity (or even short-clause DNF representations with a bounded number of satisfied clauses per input).  However, using the full statement of Theorem \ref{sav1} makes our proof slightly simpler.

We can derive from Theorem  \ref{sav1} the following form of Savick\'{y}'s result, which will be more convenient for us:

\begin{theorem}\label{sav2}  There exists a family of TUSPs

$\{W_N \subset \{0, 1, *\}^{m(N)}  \}_{N = 1} ^{\infty}$ on $m(N) = O(\poly(N))$ input bits, and a constant $\alpha > 0$, such that $D(W_N) \geq s(W_N)^{1+ \alpha}$, while  $s(W_N)\geq N$.
\end{theorem}

\begin{proof} For any $i > 0$, given uDNF representations $F_1, F_2$ of $G_i$ and $\overline{G_i}$ respectively satisfying condition (i) of Theorem \ref{sav1}, we define a search problem $V_i$:  For every clause $c$ in one of the $F_i$'s, define a witness $w_c \in V_i$ that has 0/1 entries exactly on the variables contained in $c$, with these variables set in the unique way satisfying $c$ (remember $c$ is a conjunction).  From the facts that $F_1, F_2$ are each uDNFs and that every input $x$ satisfies exactly one of them, we conclude that $V_i$ is a TUSP.  

By condition (i) of Theorem \ref{sav1}, $s(V_i) \leq 3^i$.   On the other hand, since any decision tree for $V_i$ immediately yields a decision tree of the same depth for $G_i$, we have
\[D(V_i) \geq  \frac{4^i + 2}{3} > \frac{(3^i)^{\log_3(4)}}{3},\]  
which for large enough $i$ is greater than $s(V_i)^{1 + \alpha}$ for an appropriate constant $\alpha > 0$.  Also, by Theorems \ref{quadratic} and \ref{sav1}, 
\[s(V_i) \geq \sqrt{D(V_i)} >  \frac{2^i}{\sqrt{3}} .\] 
Now we simply set $W_N = V_{\lceil \log (N) \rceil + 1}$.  We verify that $m(N) = 4^{\lceil \log (N) \rceil + 1} = O(\poly(N))$.  \end{proof}

In order to make effective use of the decision-tree depth lower bound contained in Theorem \ref{sav2}, we will need the following folklore result, showing the optimality of the `adversary method' in decision tree complexity:

\begin{claim}\label{advopt} Let $B$ be a finite set.  Suppose $f(x): \{0, 1\}^n \rightarrow B$ satisfies $D(h) \geq t > 0$; then there exists an adversary strategy for determining the bits of $x$ as they are queried (depending only on the sequence of queries made so far), such that for any query strategy making $(t - 1)$ queries to $x$, the bits of $x$ fixed in the process do not uniquely determine the value of $f(x)$.

\end{claim}

The proof of Claim \ref{advopt} is a simple inductive proof by contradiction, and is omitted.  Note that Claim \ref{advopt} applies in particular when $f(x) = W(x)$ is the (total, single-valued) function associated with a TUSP $W$.

\subsection{Set Systems and Hitting Sets} \label{setsys}

As a final preliminary definition, we introduce hitting sets of set systems, which will play a key intermediate role in the proof of Theorem \ref{main}.

Given a finite `ground set' $U$, and a collection $\mathbf{A} = \{A_1, A_2, \ldots A_l\}$ of subsets of $U$, we say a set $B \subseteq U$ \textit{hits} $\mathbf{A}$, or is a \textit{hitting set for} $\mathbf{A}$, if $B \cap A_i \neq \emptyset $ for all $i \leq l$.

Given a positive function $w: U \rightarrow \mathbb{N}$ called a `weight function', define the weight of a set $A \subseteq U$ as $w(A) = \Sigma_{u \in A}w(u)$.  Define the \textit{weighted hitting set cost} of the collection $\mathbf{A}$ (relative to $w$) as $\rho(\mathbf{A}) =$ min $\{c:$ there exists a hitting set $B \subseteq U$ for $\mathbf{A}$ with $w(B) \leq c\}$.

Given a collection $\mathbf{A} = \{A_1, \ldots A_l\}$, and given $X \in \{0, 1\}^l$, define $S_X = \{A_i: X_i = 1\}$.  Define the \textit{weighted hitting set cost function} $C_{\mathbf{A}}(X): \{0, 1\}^l \rightarrow \mathbb{N}$ by $C_{\mathbf{A}}(X) = \rho(S_X)$.

We now prove that the class of weighted hitting set cost functions is \textit{exactly} the class of economic cost functions.

\begin{lemma}\label{hsuniv} For any set system $\mathbf{A}$ and weight function $w$, $C_{\mathbf{A}}$ is an economic cost function.  Moreover, given any economic cost function $C(X): \{0, 1\}^l \rightarrow \mathbb{N}$, there exists a finite set $U$ and a collection $\mathbf{A} = \{A_1, \ldots A_l\}$ of subsets of $U$, such that for all $X \in \{0, 1\}^l$, $C_{\mathbf{A}}(X) = C(X)$.\end{lemma}

\begin{proof}  First we show that $C_{\mathbf{A}}$ is always an economic cost function. That condition (1) of the definition of economic cost functions is satisfied is immediate.  For condition (2), note that if $X \leq Y$, $S_X \subseteq S_Y$, so any hitting set for $S_Y$ is also one for $S_X$.  Thus $C_{\mathbf{A}}(X) = \rho(S_X) \leq \rho(S_Y) = C_{\mathbf{A}}(Y)$, as needed.

To see that condition (3) is satisfied, note that if $B_X, B_Y$ are hitting sets for $S_X, S_Y$, then $B_X \cup B_Y$ is a hitting set for $S_X \cup S_Y = S_{X \vee Y}$, and $w(B_X \cup B_Y) \leq w(B_X) + w(B_Y)$. 

For the second part, let $C(X): \{0, 1\}^l \rightarrow \mathbb{N}$ be an economic cost function.  We define a set system and weight function as follows.  Let $U$ be a set of size $2^l$, indexed by $l$-bit vectors as $U = \{b_X: X \in \{0, 1\}^l\}$.  

Let $\mathbf{A} = \{A_1, \ldots A_l\}$, where $A_i = \{b_X: X_i = 1\}$. Finally, define $w(b_X) = C(X)$.

We claim that, for all $X = (X_1, \ldots X_l)$, $C_{\mathbf{A}}(X) = C(X)$.  First we argue that $C_{\mathbf{A}}(X) \leq C(X)$.  Consider the singleton set $B = \{b_X\}$.  For every $i$ such that $X_i = 1$, $b_X \in A_i$.  Thus, $B$ is a hitting set for $S_X = \{A_i: X_i = 1\}$.  By definition, then, $C_{\mathbf{A}}(X) \leq w(B) = w(b_X) = C(X)$.

Now examine any hitting set $B'$ for $\{A_i: X_i = 1\}$, say $B' = \{b_{Z_{[j]}}: Z_{[j]} \in I \subseteq \{0, 1\}^l\}$.  For each $i$ such that $X_i = 1$, $A_i$ is hit by $B'$, so there exists some $Z_{[j]} \in B'$ such that $b_{Z_{[j]}} \in A_i$.  Then by definition of $A_i$, $Z_{[j] (i)} = 1$.  Thus $X \leq \bigvee_{Z_{[j]} \in I} Z_{[j]}$, and 
\[w(B') = \sum_{Z_{[j]} \in I} w(b_{Z_{[j]}}) = \sum_{Z_{[j]} \in I} C(Z_{[j]}) \geq C(\bigvee_{Z_{[j]} \in I} Z_{[j]})\]
(the last inequality holds by iterated application of property (3) of economic cost functions)
 $\geq C(X)$ (since $X \leq \bigvee_{Z_{[j]} \in I} Z_{[j]}$, and using property (2) of economic cost functions).  Thus $C_{\mathbf{A}}(X) = C(X)$, as claimed.  \end{proof}

%END OF PRELIMS

\section{Proof of Theorem \ref{main}}

\subsection{First Steps}

The first half of Theorem \ref{main} is easy, and recorded in Lemma \ref{easydir}:
\begin{lemma}\label{easydir} If $F = \{f_1(x), f_2(x), \ldots f_l(x)\}$ is a collection of nonconstant functions, $C_F(X)$ is an economic cost function.  \end{lemma}
\begin{proof} Clearly $F$ satisfies condition (1) in the definition of economic cost functions, since $C_F(\mathbf{0}) = 0$ and all decision trees computing a nonconstant function or functions has depth at least 1.

$C_F(X)$ satisfies condition (2) since, given an optimal decision tree $T$ for computing a collection $S = S_X$ of functions from $S$, and given a subset $S' = S_{X'} \subseteq S$, we can modify $T$ by removing the coordinates of its output vectors corresponding to the functions in $S \setminus S'$, yielding a decision tree $T'$ of the same depth computing the collection $S'$.  So $D(S_{X'}) \leq D(S_X)$ and $C_F(X') \leq C_F(X)$.

To show that $C_F(X)$ satisfies condition (3), let $T_X, T_Y$ be optimal decision trees of depths $d_1, d_2$ respectively, for computing the collections $X, Y$ respectively.  We define a decision tree $T'$ as follows: we replace each output node $u$ of $T_X$ with a copy $T_{Y, u}$ of $T_Y$, and on an output node $v$ of the copy $T_{Y, u}$ we place the label $(z(u), z(v))$, where $z(u)$ is the label of $u$ in $T_X$ and  $z(v)$ is the label of $v$ in $T_Y$.  Then $T'$ computes the collection $S_X \cup S_Y$ (possibly with redundant coordinates that we can remove, and up to a reordering of the outputs).  The depth of the new tree is $d_1 + d_2$. This yields condition (3).
\end{proof}

Now we turn to the second, harder half of Theorem \ref{main}.  Following Lemma \ref{hsuniv} showing the `universality' of hitting set cost functions, our approach to proving Theorem \ref{main} is to build a collection of functions mimicking the structure of a given set system $\mathbf{A}$, where each $f_i$ we create will correspond to some $A_i \in \mathbf{A}$.  We will prove:

\begin{lemma} \label{hsfuncs}  Given any hitting set cost function $C_{\mathbf{A}}(X): \{0, 1\}^l \rightarrow \mathbb{N}$ and $\epsilon > 0$, there exist integers $n$, $T$, and a collection $F = \{f_1(x), \ldots f_l(x)\}$ of functions on $n$ bits, such that, for all $X \in \{0, 1\}^l$, 
\[C_{\mathbf{A}}(X)\cdot T(1 - \epsilon) \leq C_F(X) \leq C_{\mathbf{A}}(X) \cdot T(1 + \epsilon).\]  \end{lemma}

In light of Lemmas \ref{hsuniv} and \ref{easydir}, this will prove Theorem \ref{main}.

\subsection{Bins and Mystery Bins} \label{mbinsec}

%Definition:

Central to our construction of the function family of Lemma \ref{hsfuncs} is a technical device called a `bin'.  A \textit{bin function} is a function $\mathbf{B}(y)$ mapping a Boolean input $y$ (of some fixed length) to subsets of size 0 or 1 of a set $K = \{k_1, \ldots k_M\}$ called the `keyspace'.  We call the input $y$ a `bin', and say that $k$ is `in the bin $y$' if $\mathbf{B}(y) =\{k\}$.

Our input $x$ to the function collection of Lemma \ref{hsfuncs} is going consist of disjoint bins, one bin corresponding to each $u \in U$ from our set system $\mathbf{A}$.  The bins will have different parameters; loosely speaking, we want the difficulty of determining the bin contents $\mathbf{B}_u (y_u)$ of the bin $y_u$ corresponding to $u \in U$ to be proportional to $w(u)$.  This property by itself would be relatively easy to guarantee, but we need our bins to have some other special properties as well, formalized next in the definition of `mystery bins'.

Given $\beta \in [0, 1]$ and an integer $q \geq 1$, say that $\mathbf{B}$ has \textit{security $\beta$ for $q$ queries}, and write $\secur(\mathbf{B}, q) \geq \beta$, if there exists an adversary strategy for answering queries to the vector $y$ such that, for any query strategy making $q$ queries to $y$, there exists a set $H \subset K$ of size $\beta  |K|$, such that for any key $k \in H$, the bits of $y$ fixed in the process are consistent with the condition $\mathbf{B}(y) = \{k\}$. (We do not require that the bits seen be consistent with the condition $\mathbf{B}(y) = \emptyset $, although the adversaries we will define in our construction do achieve this.)

Note that in this definition, we require an adversary strategy for deciding the input bits as they are queried, with answers depending only on the questions and answers so far, \textit{not} on the strategy/program making the queries.

Fix $T > 0$, $\delta \in (0, 1)$.  A bin function $\mathbf{B}(y)$ is called a $(T, \delta)$-\textit{mystery bin function (MBF)} (with keyspace $K$), if 

\begin{itemize}
\item[(i)] There is a $T$-query algorithm to compute $\mathbf{B}(y)$;
\item[(ii)] For any $k \in K$, it can be decided in $\delta  T$ queries whether $k \in \mathbf{B}(y)$;
\item[(iii)] $\secur(\mathbf{B}, (1 - \delta ) T) \geq (1 -\delta )$.
\end{itemize}
(Note the correspondence, when $\delta$ is close to 0, between these conditions and their informal versions in the proof sketch from the Introduction.)

Constructing mystery bin functions seems to crucially rely on a result like Theorem \ref{sav2} and its associated TUSP.  Note that mystery bin functions behave quite similarly to the TUSPs from Theorem \ref{sav2}: given a particular potential witness $w \in W$, it is easily checked if the input $x$ agrees with $w$; but computing $W(x)$ may be much harder.  The main additional ingredient in mystery bin functions is the property (iii) above, which imposes on algorithms a `sharp transition' between near-total ignorance and certainty as they attempt to determine a bin's contents.  This sharp transition is what will allow us to tightly analyze the function collections we will build to prove Lemma \ref{hsfuncs}.

Our construction of mystery bin functions is given by the following Lemma:

\begin{lemma}\label{mystbin} For all $\delta > 0$, we can find $T, M > 0$ such that, for every integer $c \geq 1$, there exists a $(c T, \delta)$-mystery bin function with keyspace $K = [M]$.

\end{lemma}

\subsection{Application of Mystery Bins}

Before proving Lemma \ref{mystbin}, we show how it is used to prove Lemma \ref{hsfuncs} and, hence, Theorem \ref{main}.

Say we are given a collection $\mathbf{A} = \{A_1, A_2, \ldots A_l\}$ of subsets of a universe $U$, and a weight function $w: U \rightarrow \mathbb{N}$.  We wish to produce a collection $F = (f_1, \ldots f_l)$ of functions such that the cost of computing a subset of the functions of $F$ is approximately a fixed scalar multiple of the minimum cost under $w$ of a hitting set for the corresponding sets in $\mathbf{A}$.

Let $w_{max}$ be the largest value of $w(u)$ over $U$.  For each $u \in U$, we define a block of input $y_u$ corresponding to $u$ and a bin function $\mathbf{B}_u$ taking $y_u$ as input.  $\mathbf{B}_u$ is chosen as a $(w(u) T, \frac{\epsilon}{  w_{max} l |U|})$-MBF with keyspace $K = [M]$, for some $T, M > 0$ independent of $u$, as guaranteed by Lemma \ref{mystbin}.  Let the input $x$ to $F$ be defined as the disjoint union of all the $y_u$.

For $i \leq l$, define $f_i(x)$ by 
\[f_i(x) = 1 \Leftrightarrow \exists k \in [M] \text{ such that }  \mathbf{B}_u (y_u) = \{k\}, \forall u \in A_i . \]

We claim that $F$ satisfies the conclusions of Lemma \ref{hsfuncs}.  If $X = \mathbf{0}$ the statement is trivial, so assume $X \neq \mathbf{0}$.  First we show the upper bound on $C_F(X)$.  Given the corresponding nonempty subset $S_X \subseteq \mathbf{A}$, let $B \subseteq U$ be a hitting set for $S_X$ of minimal cost: 

\[w(B) = \rho (S_X) = C_\mathbf{A}(X).\] 

Define an algorithm $P_X$ to compute $\{f_i (x): X_i = 1\}$ as follows: 

\begin{itemize}
\item []\textbf{Phase 1:} For each $u \in B$, compute bin contents $\mathbf{B}_u(y_u)$.
\item []\textbf{Phase 2:} For every $i$ such that $X_i = 1$, pick some $u \in B \cap A_i$ (such a $u$ must exist, since $B$ is a hitting set for $S_X$).  If in Phase 1 it was found that $\mathbf{B}_u(y_u) =\emptyset$, clearly $f_i(x) = 0$.  Otherwise, suppose $\mathbf{B}_u(y_u) = \{k\}$ for some $k \in M$; in this case, query each mystery bin $\mathbf{B}_{u'}(y_{u'})$ such that $u' \in A_i$, to ask whether $k \in \mathbf{B}_{u'}(y_{u'})$.  By the definitions, $f_i(x) = 1$ iff $k$ is indeed the contents of all such bins, so the queries of $P_X$ determine $f_i (x)$ and the output nodes of $P_X$ can be labeled to compute $f_i(x)$, for every $i$ with $X_i = 1$.
\end{itemize}

Now we bound the number of queries made by $P_X$.  In Phase 1, each individual bin contents $\mathbf{B}_u(y_u)$ can be computed in $w(u) T$ queries, by property (i) of MBFs and the definition of $\mathbf{B}_u(y_u)$.  Then altogether, at most $w(B)  T = C_\mathbf{A}(X)  T$ queries are made in this Phase.

In Phase 2, each question to a bin $y_{u'}$ asking if some $k$ is in $\mathbf{B}_{u'}(y_{u'})$ can be answered in at most 

\[  \frac{\epsilon}{  w_{max} l |U|}  (w(u')T)  \leq \frac{\epsilon T}{l |U|}   \]

queries, using property (ii) of MBFs.  Since at most $ l|U|$ such questions are asked (ranging over $(i, u')$), in total at most $\epsilon T$ such queries are made during Phase 2.  Summing over the two Phases shows that $C_F(X) \leq D(P_X) \leq (1 + \epsilon) T \cdot C_{\mathbf{A}}(X)$, as needed.

Now we show that $C_F(X) \geq (1 - \epsilon) T \cdot C_{\mathbf{A}}(X)$, again assuming $X \neq \mathbf{0}$.  We give an adversary strategy to determine the bits of $x$ as they are queried, namely:  For each $u \in U$, fix bits of $y_u$ as they're queried, by following the adversary strategy for $\mathbf{B}_u (y_u)$ given by property (iii) in the definition of MBFs, and answer queries to $y_u$ arbitrarily if this bin receives more queries than the adversary strategy for $\mathbf{B}_u (y_u)$ is guaranteed to handle.

Let $P$ be any algorithm making fewer than $(1 - \epsilon) T \cdot C_{\mathbf{A}}(X)$ queries to the input $x$; we will show that the queries made by $P$ against the adversary just defined fail to determine some value $f_i(x)$, for some $i$ such that $X_i = 1$.

For $u \in U$, let $q_u$ be the number of queries made by $P$ to $y_u$ against this adversary strategy.  Let  $B_P = \{u: q_u \geq (1 - \frac{\epsilon}{2})  w(u) T \}$ . 
We claim $B_P$ is not a hitting set for $S_X$.  To see this, note that
\[   w(B_P) = \sum_{u \in B_P} w(u)  \leq \sum_{u \in B_P} \frac{q_u}{(1 - \frac{\epsilon}{2})T}       \]
\[      \leq \frac{1} {(1 - \frac{\epsilon}{2})T}  \left(  \sum_{u \in U} q_u  \right)     \]
\[ <        \frac{1}{(1 - \frac{\epsilon}{2})T}       \left( (1 - \epsilon) T \cdot C_{\mathbf{A}}(X)   \right)    < C_{\mathbf{A}}(X), \]
so, by definition of $C_{\mathbf{A}}(X) = \rho (S_X)$, $B_P$ is not a hitting set for $S_X$.

Thus there exists an $i$ such that $X_i = 1$ and such that for every $u \in A_i$, $q_u < (1 - \frac{\epsilon}{2}) w(u) T$.  For each such $u$, by the guarantee of the adversary strategy used for bin $y_u$, there exist at least
\[(1 - \frac{\epsilon}{ w_{max}l |U|}  ) M > (1 - \frac{1}{|A_i|}) M \]
distinct keys $k \in [M]$, such that it is consistent with the bits of $y_u$ seen by $P$ that $\mathbf{B}_u(y_u) = \{k\}$.

By a union bound, there exists some fixed $k \in K$ such that it is consistent with the bits seen that $\mathbf{B}_u(y_u) = \{k\}$ for \textit{all} $u \in A_i$, which would cause $f_i(x) = 1$.  
On the other hand, it is also clearly consistent that \textit{not} all such bin contents $\mathbf{B}_u(y_u)$ are equal, and hence that $f_i(x) = 0$.  Thus $P$ fails to correctly compute $f_i(x)$, for at least one input $x$.  Since $X_i = 1$, we have shown that $C_F(X) \geq (1 - \epsilon) T \cdot C_{\mathbf{A}}(X)$.  This finishes the proof of Lemma \ref{hsfuncs}, assuming Lemma \ref{mystbin}. \qed

\subsection{Construction of Mystery Bins}

Now we prove Lemma \ref{mystbin}.  First, suppose we can prove Lemma \ref{mystbin} for $c = 1$; we'll show the conclusion then follows for every $c \in \mathbb{N}$, with the same values of $T$ and $M = |K|$.

Let $\mathbf{B}(y)$ be a $(T, \delta)$-MBF.  Say the input $y$ has length $m$; define a new bin function $\mathbf{B}_c(y')$ on input $\{0, 1\}^{cm}$ with the same keyspace $K$ by breaking the input $y'$ into $m$ blocks of size $c$, defining $z_i$ to be the sum mod 2 of the $i$th block ($i \leq m$), and setting $\mathbf{B}_c(y') :=  \mathbf{B}(z_1, \ldots z_m)$.

The adversary strategy $S'$ for $\mathbf{B}_c(y')$ is simply lifted from the strategy $S$ for $\mathbf{B}(y)$, by answering queries in any given block $i$ of $y'$ as 0s until the last, `critical' query to that $i$th block is made, then answering this query as the strategy $S$ would fix $y_i$ conditioned on the `critical' responses made so far.  Clearly any algorithm making $q$ queries can induce at most $\lfloor \frac{q}{c} \rfloor$ critical responses from the adversary, and so property (iii) in the definition of MBFs is easily seen to be inherited by $\mathbf{B}_c$.

Similarly, any algorithm for determining the bin contents $\mathbf{B}(y)$, or for querying whether $k \in \mathbf{B}(y)$ for some $k \in K$, can be adapted to $\mathbf{B}_c$ by simply querying entire blocks at a time.  This increases the number of queries by a factor $c$, giving properties (i) and (ii).  Thus $\mathbf{B}_c(y')$ is a $(c T, \delta)$-MBF with keyspace $[M]$, as needed.

\vspace{.5 em}

Now we prove Lemma \ref{mystbin} for the case $c = 1$.

Let $N > 0$ be a (large) integer to be determined later, and let $W = W_N$ be the TUSP guaranteed by Theorem \ref{sav2} for parameter $N$, with input size $m(N) = O($poly$(N))$.  For brevity write $D_N = D(W), s_N = s(W)$, and recall $D_N \geq s_N^{1 + \alpha}$, $s_N \geq N$.  We let $K := [D_N^2]$ be the keyspace.

We next describe the structure of the `bin' input $y$.  $y$ is broken into three disjoint parts, written as 
\[y = (x, \mathsf{WtK}, \mathsf{KtW}), \text{  where:}\]

\begin{itemize}
\item $x$ will be an input to the TUSP $W = W_N$ (so $|x| = m(N)$);
\item $\mathsf{WtK}$, called the `witness-to-key table', will be an encoding of a function $G_{\mathsf{WtK}}: W \rightarrow K$ (with a specific encoding method to be described shortly);
\item $\mathsf{KtW}$, called the `key-to-witness table', will be an encoding of a function $G_{\mathsf{KtW}}: K \rightarrow W$ (with a different encoding method, also described shortly).
\end{itemize}

In our definitions, every setting to the input tables $\mathsf{WtK}, \mathsf{KtW}$ will define functions $G_{\mathsf{WtK}}, G_{\mathsf{KtW}}$ as above, and such functions will generally not have unique encodings.

Assuming for now that the two encoding schemes have been fixed, we define the bin function $\mathbf{B}(y)$ as follows: $k \in \mathbf{B}(y)$ if the witness $w = W(x)$ satisfies 
\[G_{\mathsf{WtK}}(w) = k, \quad{} G_{\mathsf{KtW}}(k) = w.\]
Note that at most one key can be in the bin by this definition (or the bin may be empty).

Now we describe the encodings.  $\mathsf{KtW}$ simply uses any efficient encoding with a table entry $\mathsf{KtW}|_k$ corresponding to each element $k$ of the domain $K$.  Since each $w \in W \subset \{0, 1, *\}^{m(N)}$ has at most $s_N$ 0/1 entries, $|W|$ cannot be too large, namely

\[|W| \leq \sum_{i \leq s_N}\binom{m(N)}{i} = N^{O(s_N)},  \]

since $m(N) = O(\poly(N))$.  Thus each table entry $\mathsf{KtW}|_k$ in $\mathsf{KtW}$ can be represented using $O(s_N \log(N))$ bits.  We do so, assigning `leftover' codewords arbitrarily to elements of $W$, so that every table defines a function (and also every function is representable).

For the encoding $\mathsf{WtK}$, we want table entries to be `obfuscated', so that it takes many queries to learn anything about an individual value of $G_{\mathsf{WtK}}$.  We make the following definition, which resembles more-demanding definitions in \cite{CDHKS}:

Fix integers $m, d, t > 0$.  Say that a mapping $J: \{0, 1\}^m \rightarrow [d]$ is an $(m, d, t)$\textit{-weak Exposure-Resilient Function (wERF)} if for every $c \in [d]$ and every subset $S \subset [m]$ of size at most $t$, there is a $b \in \{0, 1\}^m$ with $J(b) = c$, and such that the entries of $b$ indexed by $S$ are all-zero.
\begin{claim}\label{werf} For sufficiently large $N > 0$, there exists a $(\lfloor s_N^{1 + \alpha / 2} \rfloor, D_N^2, \lfloor  \frac{1}{2} s_N^{1 + \alpha / 2} \rfloor)$-wERF $J$.

\end{claim}

\begin{proof} Let $J$ be a uniformly chosen random function from the domain $\{0, 1\}^m$ (with $m = \lfloor s_N^{1 + \alpha / 2} \rfloor$) to the range $[d] = [D_N^2]$. We show that with nonzero probability $J$ satisfies the definition of an $(m, d, t)$-wERF with $t := \lfloor \frac{1}{2}  s_N^{1 + \alpha / 2} \rfloor$.

Fix any subset $S \subset [m]$ of size $t$, and a $c \in [d]$.  We analyze the probability $p_{S, c}$ that there is no $b \in J^{-1}(c)$ such that $b$ is all-zero when restricted to the coordinates in $S$.  This is simply $(1 - \frac{1}{d})^{2^{m - t}}$.  Now by our settings, for sufficiently large $N$ we have $m - t \geq d 2^{m/3}$.  Thus for such $N$,
\[p_{S, c} \leq \left(1 - \frac{1}{d}\right)^{d  2^{m/3}} \leq  e^{-2^{m/3}} .\]
Taking a union bound over all choices of $S, c$, the probability that $J$ fails to be an $(m, d, t)$-wERF is, for large enough $N$, less than $2^m d e^{-2^{m/3}}  = o(1)$.  So, with nonzero probability we succeed. \end{proof}

Recall that in our setting $K = [D_N^2]$.  We let each table entry $\mathsf{WtK}|_w$ of $\mathsf{WtK}$ (with position indexed by a witness $w \in W$) contain $\lfloor s_N^{1 + \alpha / 2} \rfloor$ bits, and define

\[G_{\mathsf{WtK}}(w) = J(\mathsf{WtK}|_w),   \]
where $J$ is as given by Claim \ref{werf}.  

This completes our description of the bin function $\mathbf{B}(y)$.  We now show that for a large enough choice of $N$ it is a $((1 + \frac{\delta}{2})D_N, \delta)$ mystery bin function.

First we verify property (i) in the definition of MBFs.  In order for a query algorithm to determine the bin contents $\mathbf{B}(y)$, it suffices to do the following: Inspect $x$ to determine $w = W(x)$; look up $G_{\mathsf{WtK}}(w)$, finding some key $k$; finally, check to see if $G_{\mathsf{KtW}}(k) = w$.  If so, $\mathbf{B}(y) = \{k\}$, otherwise the bin is empty.

The first step can be implemented in $D_N$ queries to $x$.  For the second step, table entries of $\mathsf{WtK}$ are of size $\lfloor s_N^{1 + \alpha / 2} \rfloor$, which is $o(D_N)$ since $D_N \geq s_N^{1 + \alpha}$.  The third step, querying a table entry of $\mathsf{KtW}$, takes $O(s_N\log(N))$ queries, which is also $o(D_N)$.  Thus the total number of queries is $D_N(1 + o(1))$, less than $(1 + \frac{\delta}{2})D_N$ for large enough $N$.  This shows property (i).

For property (ii) of MBFs, let $k \in K$ be any key; to determine if $\{k\} = \mathbf{B}(y)$, our algorithm queries $\mathsf{KtW}|_k$ to find $w = G_{\mathsf{KtW}}(k)$ and, subsequently, queries $\mathsf{WtK}|_w$ to determine if $k = G_{\mathsf{WtK}}(w)$.  If not, then $\{k\} \neq \mathbf{B}(y)$, and the algorithm reports this.  If $k = G_{\mathsf{WtK}}(w)$, then the algorithm makes at most $s_N = o(D_N)$ queries to $x$ to see if $x$ agrees with $w$.  Note that each step takes $o(D_N)$ queries, smaller than $\delta  (1 + \frac{\delta}{2})D_N$ for large $N$.  This gives property (ii).

Finally, we show property (iii).  This is the property for which we will use the fact that $|K|$ is large and entries of $\mathsf{WtK}$ are`exposure-resilient'.  Our adversary strategy against algorithms making at most $(1 - \delta)(1 + \frac{\delta}{2}) D_N < (1 - \frac{\delta}{2}) D_N$ queries to $y$ is as follows:

\begin{itemize}
\item Answer queries to $x$ according to a strategy, guaranteed to exist by Claim \ref{advopt}, that prevents any query strategy making fewer than $D_N$ queries to $x$ from uniquely determining the value $W(x)$.  Answer all queries to $\mathsf{WtK}, \mathsf{KtW}$ with zeros.

\end{itemize}

Our proof of correctness is by contradiction.  Suppose some deterministic algorithm $P$ makes at most $(1 - \frac{\delta}{2}) D_N$ queries to $y$ against this adversary, and afterwards outputs a list $L$ of fewer than $(1 - \delta) |K|$ keys, such that the bin contents $\mathbf{B}(y)$ is forced by the bits seen to either be empty or contain a key from $L$.

Define a new algorithm $P'$ as follows: in Phase 1 $P'$ first simulates $P$ on $y$, making all the queries $P$ does.  After $P$ terminates, define $V \subseteq W$ as the set of all witnesses $w$ for which $P$ has made more than $\lfloor \frac{1}{2} s_N^{1 + \alpha / 2}  \rfloor$ queries to the table entry $\mathsf{WtK}|_w$ in $\mathsf{WtK}$.  In Phase 2, for each $w \in V$ in turn, $P'$ makes any additional queries to $x$ necessary to determine whether $x$ agrees with $w$.

Say this latter set of queries in Phase 2 are `on behalf of $w$'.  Note that for every $w \in V$, at most $s_N$ queries are made on behalf of $w$ in Phase 2, while more than $\lfloor \frac{1}{2} s_N^{1 + \alpha / 2} \rfloor$ queries are made to the table entry $\mathsf{WtK}|_w$ in Phase 1.  It follows that only an $o(1)$ fraction of the queries of $P'$ are made in Phase 2, hence for large enough $N$, $P'$ makes fewer than $D_N$ queries to $y$.

But we claim that $P'$ succeeds in determining $W(x)$, contrary to the guarantee of our adversary strategy from Claim \ref{advopt}.  First, after the simulated operation of $P$ by $P'$, say that a witness $w$ is `live' if the bits of $x$ seen are consistent with the possibility $W(x) = w$.  Note that if there is a live witness $w$ whose table entry in $\mathsf{WtK}$ has been queried at most $\lfloor \frac{1}{2} s_N^{1 + \alpha / 2} \rfloor$ times, then the value $G_{\mathsf{WtK}}(w)$ is completely undetermined (any value is consistent with the bits seen), since the adversary answered those queries with zeros and the function $J$ used in defining $G_{\mathsf{WtK}}(w)$ is a $(\lfloor s_N^{1 + \alpha / 2} \rfloor, D_N^2, \lfloor \frac{1}{2} s_N^{1 + \alpha / 2} \rfloor)$-wERF.

Thus, for any key $k$ whose table entry in $\mathsf{KtW}$ was not queried by $P$, it is consistent with the bits of $y$ seen that $\mathbf{B}(y) = \{k\}$.  Since $|K| = D_N^2  = \omega(D_N)$, if $N$ is sufficiently large then $P$ cannot query a bit from even a $\delta$ fraction of $\mathsf{KtW}$'s table entries.  Hence, for $P$ to output the list of candidates $L \subset K$ with $|L| < (1 - \delta)|K|$, it must be that for every $w \in W$ still live after the operation of $P$, the $\mathsf{WtK}$ entry for $w$ must have been queried more than $\lfloor \frac{1}{2} s_N^{1 + \alpha / 2} \rfloor$ times, and thus $w \in V$.  

Since no two distinct $w \in W$ are compatible, it follows that exactly one $w$ remains live after Phase 2 of the operation of $P'$, so $P'$ determines the value $W(x)$.  Again, this is in contradiction to the guarantee of our adversary strategy from Claim \ref{advopt}, so the assumption about $P$ was false.  We have proved that $\mathbf{B}(y)$ satisfies property (iii) in the definition of MBFs, and altogether we have shown that $\mathbf{B}(y)$ is a $((1 + \frac{\delta}{2})D_N, \delta)$-MBF, proving Lemma \ref{mystbin} for $c = 1$ (with $T = (1 + \frac{\delta}{2})D_N, M = D_N^2$).  \qed

\section{Acknowledgments}

I would like to thank Scott Aaronson and Russell Impagliazzo for their support and encouragement, Michael Forbes for valuable proofreading assistance, Brendan Juba and Shubhangi Saraf for helpful discussions, and Suresh Venkatasubramanian for pointing me to the definition of lattice outer measures in \cite{Bir}.

\end{document}